\documentclass{llncs}

\usepackage[ansinew]{inputenc}
\usepackage{graphicx,graphics,epsf,amssymb}
\usepackage{colortbl}
\usepackage{multirow}

\graphicspath{{fig/}} 

\RequirePackage{lineno}

\def\R{\mathbb{R}}
\def\eps{\varepsilon}

\newcommand{\x}[1]{x({#1})}
\newcommand{\y}[1]{y({#1})}

\newcommand{\entry}{\mathrm{En}}

\newcommand{\cen}[1]{Center({#1})}
\newcommand{\ball}[1]{B({#1})}

\begin{document}
\title{The 1-Center and 1-Highway problem}

\author{Jos\'{e} Miguel D\'{\i}az-B\'{a}\~{n}ez$^1$\fnmsep\thanks{J.M. D.-B., P. P.-L. and I. V. were partially supported by project FEDER MEC MTM2009-08652. J.M. D.-B., M.K and I. V. were partially supported by ESF EUROCORES programme EuroGIGA, CRP ComPoSe: grant EUI-EURC-2011-4306.},
Matias Korman$^2$\fnmsep\thanks{Partially supported by the support of the Secretary for Universities and Research of the Ministry of Economy and Knowledge of the Government of Catalonia and the European Union.},
Pablo P\'erez-Lantero$^{3}$\thanks{Partially supported by grant FONDECYT 11110069.},
Inmaculada Ventura$^{1}$\fnmsep$^*$}

\institute{ 
Departamento de Matem\'{a}tica Aplicada II, Universidad de Sevilla, Spain. {\tt \{dbanez,iventura\}@us.es}.
\and
Universitat Polit\`ecnica de Catalunya (UPC), Barcelona. {\tt matias.korman@upc.edu}.
\and
Escuela de Ingenier\'{i}a Civil en Inform\'{a}tica, Universidad de Valpara\'{i}so, Chile. {\tt pablo.perez@uv.cl}. }

%

\maketitle

\begin{abstract}
We study a variation of the $1$-center problem in which, in addition to a single supply facility, we are allowed to locate a highway. This highway increases the transportation speed between any demand point and the facility. That is, given a set $S$ of points and $v>1$, we are interested in locating the facility point $f$ and the highway $h$ that minimize the expression $\max_{p\in S}d_{h}(p,f)$, where $d_h$ is the time needed to travel between $p$ and $f$. We consider two types of highways ({\em freeways} and {\em turnpikes}) and study the cases in which the highway's length is fixed by the user (or can be modified to further decrease the transportation time).
%
\end{abstract}

\textit{Keywords:} Geometric optimization; Facility location; Time metric.

\section{Introduction}
The optimal location of a facility modeled as a geometric object is a well-studied problem both in operations research and computational geometry. An overview on continuous location is given in \cite{plastria}. Particularly, the geometrical nature of problems under the minmax criterion has led to fruitful interaction between both fields \cite{hamacher,godfried}. On the other hand, models dealing with alternative transportation systems have been suggested in location theory \cite{mesa}.

Although the metric given by a real urban transportation system is often quite complicated, simplified mathematical models have been widely studied in order to investigate basic geometric properties of urban transportation systems. Abellanas {\em et. al.}~\cite{ahiklmps-vdsnh-03} considered a geometric modeling of this environment: represent highways as line segments in the plane, giving each line segment an associated speed. Then, the travel time between two points gives a metric called the {\em city distance}. Recently, there has been an interest in facility location problems derived from urban modeling. In many cases we are interested in locating a highway that optimizes some given function that depends on the distance between elements of a given pointset (see for example \cite{cardinal08,kt-oishcm-08}).

In a recent paper,  Espejo and Ch\'{i}a~\cite{espejo11} introduced a variant of the problem in which we are given a set of clients  (represented by a set of points $S$) located in a city and one is interested in locating a  service facility and a highway simultaneously in a way that the average supply time between the clients and the supply point is minimized.
The state of the art and applications of this problem can be found in their paper and the references therein.
Unfortunately, it was shown that their algorithm could give an incorrect solution in some cases and a new corrected algorithm is given in~\cite{bklv-lsfrtl-11}. In this paper we study a variation of this problem in which, instead of the average travel time, we want to minimize the largest travel time between the clients and the facility.

\subsection{Definitions and notation}
Let $S$ be the set of $n$ client points, $f$ be the service facility point, $h$ be the highway, represented by a segment whose endpoints are $t$ and $t'$. Given a point $u$ of the plane, let $\x{u}$ and $\y{u}$ denote respectively the $x$ and $y$ coordinates of $u$. For simplicity in the explanation, we assume that no two point of $S$ share an $x$ or $y$ coordinate, but we note that our algorithms work for a the case in which points do not have this general position assumption.  Also, let $\ell$ be the length of $h$, and $v>1$ be its speed. 

\begin{figure}[h]
    \centering
    \includegraphics[width=0.9\textwidth]{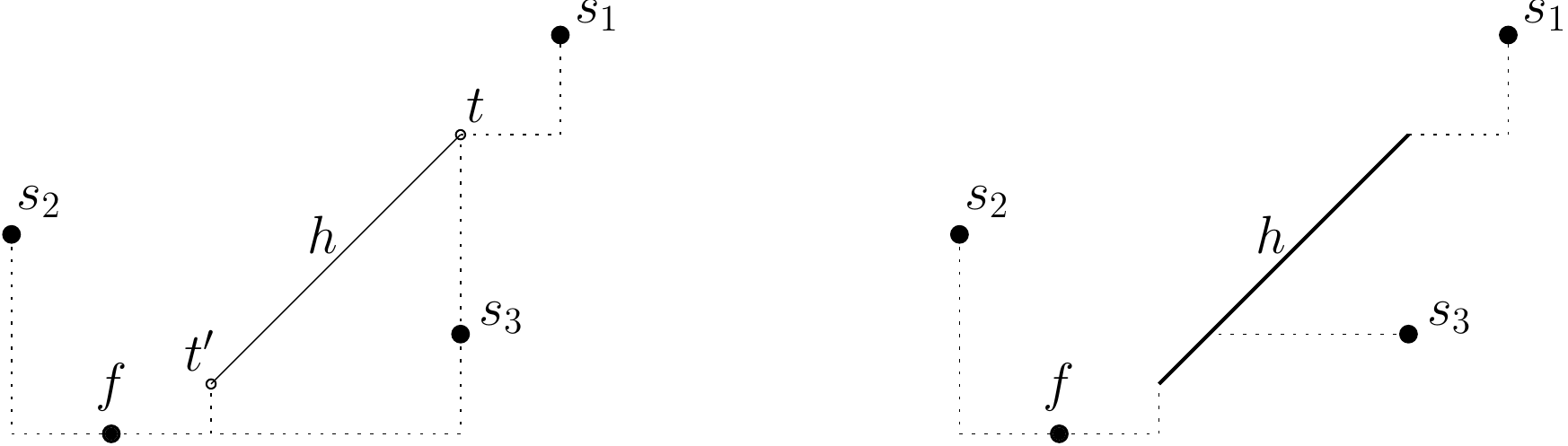}
    \caption{\small{The distance model using both turnpikes (left, denoted by a thin segment and its two endpoints) and freeways (right, denoted by a thick segment); in the first case $s_1$ uses the turnpike from $t$ to $t'$ in order to reach $f$ faster. The turnpike does not speed up transportation between $s_2$ and $f$, hence is not used by $s_2$. Demand point $s_3$ however, can either walk or use the turnpike to reach $f$, and will need the same time in both cases. In the right we have the exact same instance, but the highway is now a freeway. Observe that in this case, point $s_3$ will take profit from the freeway.}}
    \label{fig:distance}
\end{figure}

Two different types of highways have been considered in the literature: {\em freeways} \cite{aap-qpsscvd-04,kt-oishcm-08} and {\em turnpikes} \cite{cardinal08,bae09}. The difference between a freeway and a turnpike is that we can only enter or exit from a turnpike at its endpoints, while we can enter and exit from a highway from any point on it. We will use the term highway to refer to either type.

Fixed the location of the facility and the endpoints $t,t'$ of the turnpike, the distance from a demand point $p\in S$ to $f$ is defined as $d_h(p,f)=\min\{\|p-f\|_1,\|p-t\|_1+\frac{\ell}{v}+\|t'-f\|_1, \|p-t'\|_1+\frac{\ell}{v}+\|t-f\|_1\}$, where $||\cdot||_1$ represents the $L_1$ distance between two points. That is, the minimum time between either walking, using the highway in one direction or using it in the reverse direction. When the highway is a freeway, the distance is expressed as $d_h(p,f)=\min\{\|p-f\|_1,\min_{s,s'\in h} \|p-s\|_1+\frac{\|s-s'\|_1}{v}+\|s'-f\|_1\}$ instead, see Figure \ref{fig:distance}. Whenever $d_h(p,f)<\|p-f\|_1$, we say that $p$ uses the highway to reach $f$. Otherwise, we say that $p$ walks (or does not use $h$) to reach the facility.


The problem we study can be formulated as the 
\begin{quote}
{\bf The 1-Center and 1-Highway problem (1C1H problem)}: Given a set $S$ of $n$ points and a fixed speed $v>1$, locate a point (facility) $f$ and a segment highway $h$ of any orientation whose endpoints are $t$ and $t'$ such that the function $\max_{p\in S}d_{h}(p,f)$ is minimized.
\end{quote}

We consider the variants in which the highway's length $\ell=||t-t'||_2$ is fixed (and we call it the {\em fixed length} problem) or the highway can have any length ({\em variable length}). We will also consider the cases in which the highway to locate can be either a turnpike or a freeway. Thus,  in total we have four variants of the problem which we denote by FL-1C1F (fixed length 1-center 1-freeway problem), VL-1C1T (variable length 1-center 1-turnpike problem), and so on. In this paper we propose efficient algorithms to solve the four variants of the problem. 


\section{Locating a Turnpike}

In this section, we give a general algorithm to solve both variants of the 1C1H problem.
We will then propose an improved method for the VL-1C1H problem.

\medskip

It is easy to see that, when locating a turnpike, the highway will only be used in one direction.
By using similar arguments to those given in ~\cite{espejo11} (Lemma 2.1) we can prove that there always exists an optimal location in which one of the turnpike  endpoints coincides with the facility.
Therefore, throughout this section we assume that $f=t'$ thus the distance from a demand point $p\in S$ to $f$ is now $d_h(p,f)=\min\{\|p-f\|_1,\|p-t\|_1+\frac{\ell}{v}\}$.

\medskip
Using the standard transformation from $L_1$ to $L_{\infty}$, we solve the problem using $L_{\infty}$ instead.
Let $f^*$ and $h^*$ be an optimal solution of a given problem instance. Let $t^*$ be the endpoint of $h^*$ other than $f^*$ and let $R^*=\max_{p\in S}d_{h^*}(p,f^*)$.
Let $r_1$ be the maximum of $d_{h^*}(p,f^*)$ among all points $p\in S$ not using $h^*$,
and $r_2$ be the maximum of $d_{h^*}(p,f^*)-\ell/v$ among all points $p\in S$ that use $h^*$.
Let $\ball{u,r}$ denote the ball of center $u$ and radius $r$ with respect to the $L_{\infty}$-metric.
Note that $S$ is covered by the balls $\ball{f^*,r_1}$ and $\ball{t^*,r_2}$, and
$R^*=\max\{r_1,r_2+\ell/v\}$ is satisfied.
Furthermore, either $r_1$ or $r_2$ can be increased, without affecting the value $R^*$ of the solution, so that
$r_1=r_2+\ell/v$.
From these observations, the following
statement can be obtained: The 1C1H problem is equivalent to finding two balls, $\ball{f,R}$ and $\ball{t,R-\ell/v}$, such that $R$ is minimum and $\ball{f,R}\cup\ball{t,R-\ell/v}$ covers $S$.
%


%

%

\medskip

We partition the pointset $S$ into two sets $W^*$ and $H^*$ as follows: set $W^*$ contains the points whose $L_\infty$ distance to $f^*$ is at most $R^*$. The set $H^*=S\setminus W^*$ contains the points that must use the highway to reach $f^*$ in $R^*$ or less units of time.

Observe that we cannot have $W^*=\emptyset$, since by reversing the positions of $f^*$ and $t^*$ we would obtain a better solution. Notice that the set $H^*$ can be empty (for example, if we are forced to locate an extremely long highway). However, this case can be easily detected and treated, since $f^*$ is the solution of the rectilinear 1-center problem~\cite{drezner87} (which can be computed in linear time).
 Hence, from now on we assume that neither $W^*$ nor $H^*$ is empty.

We consider the next problem called the \emph{basic problem}: Given a partition $\{W,H\}$ of $S$, find the smallest value $R$ (called the {\em radius} of the partition) and the coordinates of $f$ and $t$ such that $W\subseteq \ball{f,R}$ and $H\subseteq \ball{t,R-\ell/v}$. When we consider the fixed-length variation of the problem, we also add the constraint that $f$ and $t$ must satisfy $\|f-t\|_2=\ell$. Since $f^*$ and $t^*$ are optimal, it is easy to see that they are the solution of the basic problem for the partition $\{W^*,H^*\}$. Moreover, the radius of any other partition of $S$ will have equal or higher radius than $R^*$.

Our algorithm works as follows: we consider different partitions of $S$ and solve the basic problem associated to each partition. We identify $\{W^*,H^*\}$ as the partition whose radius is smallest. A naive method would be to guess the partition $\{W^*,H^*\}$ among the $O(2^n)$ candidates. In the following we reduce the search space to one of polynomial size:

\begin{figure}[h]
    \centering
    \includegraphics[width=0.9\textwidth]{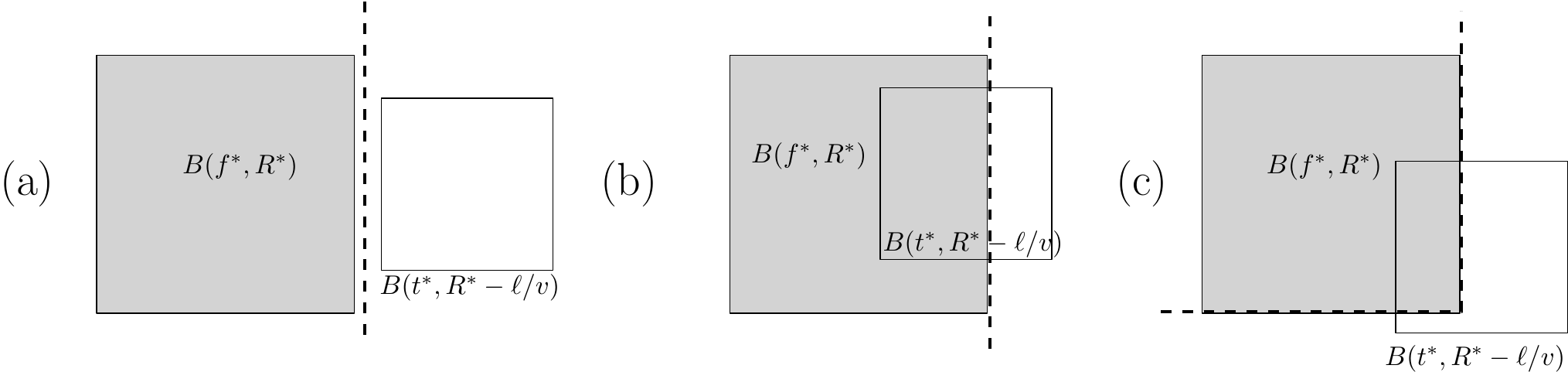}
    \caption{\small{Relative positions of the balls $\ball{f^*,R^*}$ and $\ball{t^*,R^*-\ell/v}$. For each of the cases the sets $W^*$ (marked in grey) can be split from $H^*$ with either an axis-aligned line or an upper-left quadrant.}}
    \label{fig:cases}
\end{figure}
\begin{lemma}\label{lem_key}
For any set $S$, the partition $\{W^*,H^*\}$ can be found among $O(n^2)$ candidates
\end{lemma}
\begin{proof}
Without loss of generality we can assume that $f^*$ is above and to the left of $t^*$. It is then easy to see that there are three possible relative positions of the two balls. In each of these cases the two sets can be split by either an axis-aligned line ( cases (a) and (b) in Figure~\ref{fig:cases}) or an upper left quadrant (case (c) in  Figure~\ref{fig:cases}).
%
%
Each possible partition is uniquely determined by the points belonging to the closed left (resp. North-West) halfplane (resp. quadrant)  and the remaining points.
Since there are $O(n^2)$ different cases, hence the Lemma is shown.
\qed

\end{proof}

Given a set $T$ of points, let  $X(T)\subseteq T$ be the set containing the points with highest and lowest $x$ and $y$ coordinate of $T$ (this set is called the set of {\em extreme points} of $T$). We define $\delta(T)$ as one half of the largest $L_\infty$ distance between any two points of $T$. That is, the  minimum radius needed so that all points of $T$ can be included in some $L_{\infty}$ ball. Observe that $|X(T)|\leq 4$ and that $\delta(T)=\delta(X(T))$.  For any real number $r>0$, let $\cen{T,r}$ be the locus of the centers of the axis-parallel squares of radius $r$ (i.e., $L_\infty$ balls) that cover $T$. A similar definition can be seen in~\cite{huang03}.

\begin{lemma}\label{lem_propcen}
For any set $S$ of points the following properties hold:
\begin{itemize}
    \item $\cen{S,r}$ can be empty, a point, an axis-parallel segment, or a rectangle.
    \item $\cen{S,r}=\emptyset$ if and only if $r<\delta(S)$.
    \item $\cen{S,r}=\cen{X(S),r}$. 
    \item For any $\eps>0$ and $r\geq \delta(S)$, $\cen{S,r}\subset\cen{S,r+\eps}$. Moreover, the separation between the boundaries of $\cen{S,r}$ and $\cen{S,r+\eps}$ is  $\eps$.
\end{itemize}
\end{lemma}
All these properties are easy to prove from the definition of $\cen{S,r}$ and linearity of $L_\infty$, thus we omit them. With these observations we can now solve the basic problem efficiently:

\begin{lemma}\label{lemma:basic-cte}
Let $\{W,H\}$ be a partition of $S$. If we are given the extreme sets $X(W)$ and $X(H)$, then the basic problem can be solved in constant time (for both fixed-length and variable-length cases).
\end{lemma}
\begin{proof}
We start by giving the algorithm for the case in which the distance between $f$ and $t$ is fixed. Observe that the radius $R$ of the partition will always satisfy $R\geq \max\{\delta(W), \delta(H)+\ell/v\}$ (otherwise an extreme point of either $W$ or $H$ will not be able to reach $f$ in $R$ units of time). If there exist two points $f\in \cen{W,\max\{\delta(W),\delta(H)+\ell/v\}}$ and $t\in \cen{H,\max\{\delta(W), \delta(H)+\ell/v\}}$ such that $||f-t||_2=\ell$ we are done.

Unfortunately, this does not always happen. In general, we must find two values $\eps_1,\eps_2\geq 0$ such that: $(i)$ $\delta(W)+\eps_1=\delta(H)+\eps_2+\frac{\ell}{v}$, $(ii)$ there are points $f\in \cen{W,\delta(W)+\eps_1}$ and $t\in \cen{H,\delta(H)+\eps_2}$ satisfying $\|f-t\|_2=\ell$, and $(iii)$ $\delta(W)+\eps_1=\delta(H)+\eps_2+\frac{\ell}{v}$ is minimized. The values $\eps_1$ and $\eps_2$ can be found in constant time as follows:

First, set $\eps_1=\max\{0,\delta(H)+\frac{\ell}{v}-\delta(W)\}$ and $\eps_2=\max\{0,\delta(W)-\delta(H)-\frac{\ell}{v}\}$. That is, we either increase the radius of $W$ to $\delta(H)+\ell/v$ or increase the radius of $H$ to $\delta(W)$. By increasing the smallest radius we are ensuring that condition $(i)$ is satisfied. We now look for the smallest value of $x>0$ such that there are points $f\in \cen{W,\delta(W)+\eps_1+x}$ and $t\in \cen{H,\delta(H)+\eps_2+x}$ satisfying $\|f-t\|_2=\ell$. Observe that this problem is of constant size and can be computed in $O(1)$ time. 

The variable-length variant of the problem is slightly simpler.
The main difference is that $\eps_1$ and $\eps_2$ must now minimize the expression $\max\{\delta(S_1)+\eps_1,$ $\delta(S_2)+\eps_2+g(\eps_1,\eps_2)/v\}$, where $g(\eps_1,\eps_2)$ denotes the smallest Euclidean distance between points $f\in\cen{S_1,\delta(S_1)+\eps_1}$ and a point $t\in\cen{S_2,\delta(S_2)+\eps_2}$. As before, this problem has constant size and thus can be solved in $O(1)$ time. \qed

\end{proof}

By combining the above results we obtain a method to solve both problems:

\begin{theorem}\label{theo_FL1C1T}
Both variants of the 1C1T problem can be solved in $O(n^2)$ time and $O(n)$ space.
\end{theorem}
\begin{proof}
Recall that, by Lemma \ref{lem_key}, we can split $S$ into sets $H^*$ and $W^*$ by either a  vertical line or an upper-left quadrant. We start by considering first the case in which there is a vertical splitting line, cases (a) and (b) in Figure~\ref{fig:cases}. Sort the points of $S$ in increasing value of $x$ coordinates; let $p_1, p_2, \ldots , p_n$ be the obtained order. For any $1\leq i< n$, let $L_i$ and $R_i$ be the smallest bounding axis-aligned rectangle containing points $\{p_1,\dots,p_i\}$ and $\{p_{i+1},\dots,p_n\}$, respectively. By scanning from left to right, we can compute and store the extreme points of $L_i$ for all $1\leq i<n$ in $O(n)$ time. Analogously we sweep from right to left and compute $X(R_i)$ in linear time as well. Then we solve the basic problem for each pair $(L_i,R_i)$ for all values $1\leq i<n$ using Lemma \ref{lemma:basic-cte}. The computationally speaking most expensive part of the algorithm is computing the initial sorting of the points of $S$, which needs $O(n\log n)$ time.

To complete the proof it remains to show how case (c)  in Figure~\ref{fig:cases} can be solved in $O(n^2)$ time and $O(n)$ space.
Let $q_1, q_2, \ldots, q_n$ be the points of $S$ sorted in decreasing order of $y$ coordinates. For any $1 \leq i,j,\leq n$, let $UR_{i,j}$ and $DL_{i,j}$ be the smallest bounding rectangles of the sets $UR_{i,j}:=\{u\in S~|~\y{u}>\y{q_i}\wedge \x{u}<\x{p_j}\}$ and $S\setminus UR_{i,j}$, respectively. For any fixed $1\leq i \leq n$ we can proceed as in the previous case. That is, sweep twice the pointset (downwards and upwards) and compute the set of extreme points $X(UR_{i,j})$ for all $1\leq j\leq n$. Once the bounding rectangles are known, we can solve the $O(n)$ basic problem instances in constant time each. We repeat this algorithm for all values of $i$ and find the optimal partition in quadratic time. Observe that this method never uses more than $O(n)$ memory, since once a column has been scanned we need only store the best partition found. \qed
\end{proof}

\subsection{Locating a turnpike of variable length}
Using Theorem \ref{theo_FL1C1T} we have an algorithm that runs in $O(n^2)$ time for both the fixed-length and variable-length variants of the problem. The bottleneck of the algorithm is case (c) of Lemma \ref{lem_key}. In the following we show how to treat this case more efficiently for the variable-length case.

Given $S$, let $p_N$ and $p_S$ be the points with highest and lowest $y$-coordinate, respectively. Analogously, $p_E$ and $p_W$ are defined with respect to the $x$-coordinates.
Without loss of generality, we assume that $\delta(S)=\x{p_E}-\x{p_W}$ (that is, the width of $S$ is larger than its height).
%

\begin{lemma}\label{lemma:corner}
If in every optimal solution of the VL-1C1T problem each ball contains a corner of the other one, there exists an optimal solution $(f^*,t^*)$ of radius $R^*$ of the VL-1C1T problem in which the extreme points $X(S)$ are in the boundary of $\ball{f^*,R^*} \cup \ball{t^*,R^*-\ell/v}$.
\end{lemma}
\begin{proof}
%
Observe that if any of the balls contains both $p_N$ and $p_S$ (or $p_W$ and $p_E$), the sets $H^*$ and $W^*$ can be split by either a horizontal or vertical line. Since we assumed that this is not possible, the ball that contains $p_N$ cannot contain $p_S$ (analogously, the ball that contains $p_W$ cannot contain $p_E$). Without loss of generality, we assume that $p_N,p_W \in\ball{f^*,R^*}$ and $p_S,p_E\in\ball{t^*,R^*-\ell/v}$. Observe that this implies that $\x{f^*}\leq\x{t^*}$, and $\y{f^*}\geq\y{t^*}$ (that is, the highway and the abscissa form an angle between $0$ and $-\pi/2$).

Assume that point $p_N$ is not on the boundary of $\ball{f^*,R^*}$. In the following we will show how to do a local perturbation to the solution so that we obtain $p_N$ on the boundary. We translate $f^*$ downwards continuously while keeping $t^*$ unchanged  until $p_N$ reaches the top boundary. Observe that, since initially $f^*$ has higher coordinate than $t^*$, the translation will reduce the distance between $f^*$ and $t^*$ until both points share the same $y$ coordinate. However, observe that this cannot happen, since otherwise we can split $W^*$ and $H^*$ with a vertical line. Moreover, no point of $\ball{f^*,R^*}$ can leave the ball before $p_N$ reaches the top boundary, hence optimality is preserved through this translation operation. Analogously we can do the same operation on the other extreme points and obtain that either all extreme points are in the boundary of  $\ball{f^*,R^*} \cup \ball{t^*,R^*-\ell/v}$ or find a way to split both $W^*$ and $H^*$ with a vertical or horizontal line. \qed
\end{proof}

With this observation we can speed up the algorithm for the variable length variant of the problem:

\begin{theorem}\label{theo_vl1C1T}
The VL-1C1T problem can be solved in $O(n\log n)$ time and $O(n)$ space.
\end{theorem}
\begin{proof}
By Lemma \ref{lemma:corner} either there exists an axis-parallel line that splits the sets $W^*$ and $H^*$ or all extreme points are in the boundary of $\ball{f^*,R^*} \cup \ball{t^*,R^*-\ell/v}$. The first case can be treated in  $O(n\log n)$ time using the same approach as in Theorem \ref{theo_FL1C1T}, hence we focus on the latter case.

Without loss of generality, we assume that $\x{f^*}<\x{t^*}$, and $\y{f^*}\geq\y{t^*}$.
%
%
In particular, this implies that $p_N, p_W \in B(f^*,R^*)$ and $p_E, p_S \in B(t^*,R^*-\ell/v)$. Denote by $u=(x(p_W),y(p_N))$ the top-left corner of the smallest enclosing axis-aligned rectangle of $S$. By Lemma \ref{lemma:corner} this point must also be the top-left corner of $B(f^*,R^*)$. Let $r_1, r_2,\dots, r_n$ be the elements of $S$ sorted in increasing $L_\infty$ distance to $u$. Then now apply the same approach as used in cases (a) and (b) using the new ordering instead. The result thus follows. \qed
\end{proof}
\section{Freeway location}

In this section we consider the case in which the highway to locate is a freeway (instead of a turnpike).
Now, the travel time between a demand point $p$ and the
service facility $f$, denoted by $d_h(p,f)$, is equal to:
\begin{equation}\label{eq1}
\min\left\{
\begin{array}{l}
\|p-f\|_1,\\
\min_{q_1,q_2\in
h}\left\{\|p-q_1\|_1+\frac{\|q_1-q_2\|_2}{v}+\|q_2-f\|_1\right\}
\end{array}
\right.
\end{equation}


\begin{figure}[h]
    \centering
    \includegraphics[width=0.5\textwidth]{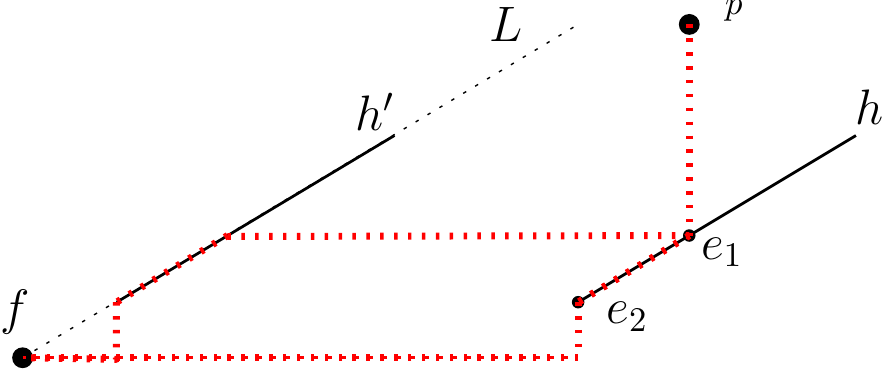}
    \caption{Illustration of the first step of the proof of Lemma \ref{lem:characfree}}
    \label{fig_impro}
\end{figure}
\begin{lemma}\label{lem:characfree}
In both variants(fixed and variable length) of the 1C1F problem  there exists an optimal solution in which the facility point is located on the highway.
\end{lemma}
\begin{proof}
Let $(f,h)$ be an optimal solution and assume that $h$ does not pass through $f$. Let $L$ be the line that passes through $f$ and has the same orientation as $h$. Translate $h$ horizontally until it passes through $L$, let $h'$ be the obtained highway. In the following we show that the distance between any point $p\in S$ cannot have increased in the translation process: if $p$ does not use highway $h$, clearly we have $d_h(p,f)\leq d_{h'}(p,f)$. Otherwise $p$ enters highway $h$ at point $e_1\in h$ (and leaves $h$ at $e_2$). Observe that the length of the path is the sum of lengths of the vectors $\vec{pe_1}$,$\vec{e_1e_2}$ and $\vec{e_2f}$, where the first and last  vectors are measured with the $L_1$ metric (and the second one is inside the freeway). Note, however that the same path can be done on $h'$ (see Figure \ref{fig_impro}), hence the travel time is unaffected.

If $f\in h'$ we are done, otherwise, translate $h'$ along the line $L$ until it hits $f$ (let $h''$ be the new highway). Analogously to the previous case, we have $d_{h'}(f,s)=d_{h''}(f,s)$, hence this second location is also optimal and the result thus follows.\qed
\end{proof}


From this point forward we consider that the facility point $f$ is in $h$. Given a potential solution $(f,h)$, let $e$ be the lower leftmost endpoint of $h$ (and $e'$ be the other endpoint). In the following we will characterize the shortest paths from any point $p\in\R^2$ to $f$. Let $\pi(p)$ be the shortest path connecting $p$ and $f$. It is easy to see that if $\pi(p)$ enters the freeway, it will not leave it until it reaches $f$. Let $\entry(p)$ be the point in which the paths enters the freeway (if the highway is not used in $\pi$ we simply define $\entry(p)=f$). Observe that $d_h(p,f)=\|p-\entry(p)\|_1+\frac{\|\entry(p)-f\|_1}{v}$.

Let $\alpha$ always denote the non-negative angle of the highway
with respect to the positive direction of the $x$-axis. Unless
otherwise specified, we assume $0\leq\alpha\leq\frac{\pi}{4}$.
Observe that if $\alpha>\frac{\pi}{4}$ we can, by properties of
$L_1$ and $L_2$ metrics, modify the coordinate system so that angle
$\alpha$ satisfies $0\leq\alpha\leq\frac{\pi}{4}$.
For any point $p\in S$ let $p'$ be the intersection point between $h$ and the vertical line passing through $p$ (if it exists). This point is called the {\em vertical projection} of $p$. Analogously we define $p''$ as the horizontal projection. We say that a point $p$ is above $h$ if it is above the line passing through $h$. Notice from the
assumption $0\leq\alpha\leq\frac{\pi}{4}$ that given $h$ and a
demand point $p$, $p'$ is the nearest point to $p$ on $h$ under the
$L_1$ metric.

The next lemma characterizes the way in which demand points move
optimally to the facility. A detailed proof is given in \cite{bklv-lsfrtlovl-11} for the case in which the length is a variable. The proof can be easily adapted for the fixed length case.
Let $\varphi_v=\frac{\pi}{4}-\arcsin\left(\frac{\sqrt{2}}{2v}\right)$. Since $v>1$
we have $0<\varphi_v<\frac{\pi}{4}$.


\begin{lemma}\label{lem:spm}\cite{bklv-lsfrtlovl-11}
If $\varphi_v<\alpha\leq\frac{\pi}{4}$, the shortest path between $p$ and $f$ satisfies one of the following:
\begin{enumerate}
\item If $x(p)\leq x(e)$,  $\entry(p)=e$.
\item If $x(p)\geq x(e')$, $\entry(p)=e'$.
\item Otherwise, $\entry(p)=p'$.
\end{enumerate}
Otherwise,  $0\leq\alpha\leq\varphi_v$ 
 and the shortest path satisfies:
\begin{enumerate}
\item if $x(p)\leq x(e)$ and $y(p)\leq y(e)$, $\entry(p)=e$.
\item if $x(p)\geq x(e')$ and $y(p)\geq y(e')$, $\entry(p)=e'$.
\item if $x(f)\leq x(p)\leq x(e')$ and $s$ is above $h$, $\entry(p)=p'$.
\item if $y(f^*)\leq y(p)\leq y(e')$ and $p$ is below $h$, $\entry(p)=p''$.
\item if $y(e)\leq y(p)\leq y(f)$ and $p$ is above $h$, $\entry(p)=p''$.
\item if $x(e)\leq x(p)\leq x(f)$ and $p$ is below $h$, $\entry(p)=p'$.
\item if $x(p)\leq x(f)$ and $y(p)\geq y(f)$, $\entry(p)=f$.
\item if $x(p)\geq x(f)$ and $y(p)\leq y(f)$, $\entry(p)=f$.
\end{enumerate}
\end{lemma}
\begin{corollary}\label{cor_convex}
For any optimal solution $(f,h)$ in which $f\in h$ and $R>0$, the ball (with respect to the metric $d_h$) of radius $R$  centered at $f$  is a convex polygon with at most eight faces.
\end{corollary}
\begin{figure}
\centering
\includegraphics[width=0.80\textwidth]{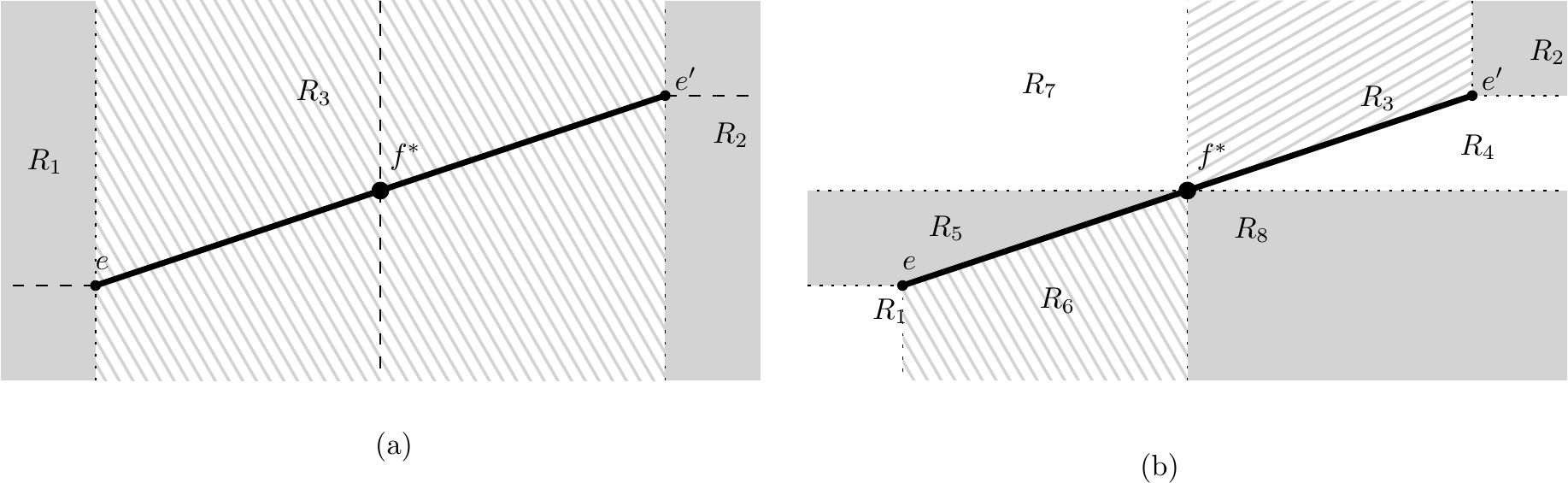}
\caption{Shortest path map of $p$. Case $(a)$ corresponds to the case in which $\varphi_v<\alpha\leq\frac{\pi}{4}$ (and $(b)$ when $0\leq\alpha\leq\varphi_v$). In each of the cases, the space has been partitioned into regions such that the shortest path of any two points in the same region is equivalent. Region $R_i$ corresponds to the case in which the $i$-th rule of Lemma \ref{lem:spm} is used, respectively. The dashed lines in (a) depict the refinement needed so as to certify that the distance function to $f$ is bivariate linear in each region.}
\label{fig_spm}
\end{figure}

Figure \ref{fig_spm} illustrates Lemma \ref{lem:spm}.
Observe that if $0\leq\alpha\leq\varphi_v=\pi/4-\arcsin(\frac{\sqrt{2}}{2v})$, the distance between $f$ and a point in a fixed region $R_i$ ($i\leq 8$) is a bivariate linear function whose coefficients of this function only only depend on $v$ and the slope of $h$. However, the same result does not occur when $\varphi_v<\alpha\leq\frac{\pi}{4}$. This is due to the fact that the partition into regions does not distinguish between the cases in which $p$ is above/below or to the left/right of $f$. For simplicity in the explanation, from now on we further refine the three regions into eight regions so that the distance in each sub-region also is a bivariate piecewise linear function. This can be done by adding a vertical line passing through $f^*$ and two horizontal rays emanating outwards from  $e$ and $e'$, see Figure \ref{fig_spm}(a).

That is, in either case, the plane can be partitioned into eight regions $R_i$ such that, for any $i\leq 8$, the distance between $f$ and a point in $R_i$ is a bivariate linear function. Given a slope $\alpha\in[0,\pi/2]$ and $i\leq 8$, let $a_i,b_i,$ be $x$ and $y$ coefficients of the distance function between the points in region $R_i$ and $f$. Also, let $c_i$ be the additive term of the distance function (that is, $d_h(p,f)=a_ix(p)+b_iy(p)+c_i$ for all $p\in \R_i$).

We extend these functions to the plane and say that $p \in S$ is the $i$-th {\em extreme} point if $p$ maximizes $a_ix(p)+b_iy(p)+c_i$ among all points of $S$ (pick any arbitrarily if many exist). We also define $E_{\alpha}$ as the set of $i$-extremes (for all $i\leq 8$). Observe that $2\leq |E_{\alpha}|\leq 8$. In general we need not have that the $i$-th extreme belongs to the region $R_i$. However, we will show that we can ignore non-extreme points:

\begin{figure}[h]
    \centering
    \includegraphics[width=0.5\textwidth]{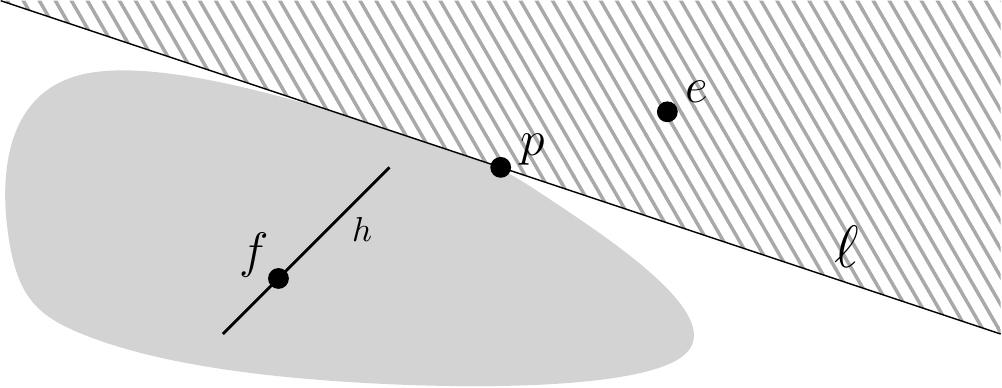}
    \caption{Proof of Lemma \ref{lem_freekey}: point $e$ cannot be contained in the interior of ball $B_h(f,d_h(f,p))$ (depicted in gray) and in the halfplane $\ell^+$ (dashed).}
    \label{fig_convex}
\end{figure}
\begin{lemma}\label{lem_freekey}
Let $(f,h)$ be an optimal solution of the 1C1F-problem for a set $S$ of points, and let ${\alpha}$ be the slope of $h$. The location of $f$ and $h$ only depends on $E_{\alpha}$.
\end{lemma}
\begin{proof}
For any $R>0$, let $B_h(f,R)$ denote the ball (with respect to the metric induced with highway $h$) centered at $f$ of radius $R$. In the following we show that if $E_{\alpha} \subseteq B_h(f,R)$ for some $R>0$, we also have $S\subseteq B_h(f,R)$. Thus, if we know the orientation of $h$ beforehand we can disregard all non-extreme points of $S$.
Let $p\in S$ be a non-extreme point, $i$ be the region in which $p$ belongs to, and $e$ be the $i$-th extreme of $S$. 
Assume that there exists some $R>0$ such that $e\in B_h(f,R)$ but $p\not\in B_h(f,R)$. By definition we have $R\leq d_h(f,p)$, so we increase the radius of the ball until $p$ is on the boundary. As a result, $e$ must be in the interior of $B_h(f,d_h(f,p))$. Let $c\geq 0$ be the value such that the line $\ell=\{(x,y)| a_ix+b_iy+c=0\}$ passes through $p$. We also define $\ell^{-}=\{(x,y)| a_ix+b_iy \leq c\}$ as the halfplane containing the points on or below $\ell$, and $\ell^{+}=\{(x,y)| a_ix+b_iy \geq c\}$ as the opposite halfplane. Without loss of generality, we can assume that $f$ is in $\ell^-$. Observe that, since $e$ is the extreme point with respect to $a_i$ and $b_i$, we have $e\in\ell^+$ (otherwise $p$ would be extreme instead).

By Corollary \ref{cor_convex}, the ball $B_h(f,d_h(f,p))$ is convex. More importantly, since $p\in R_i$, the line tangent to $B_h(f,d_h(f,p))$ at $p$ must be  line $\ell$. In particular, the ball $B_h(f,d_h(f,p))$ is contained in the halfplane $\ell^-$. However, we have a contradiction, since $e$ cannot be an interior point to $B_h(f,d_h(f,p))$ and on $\ell^+$ (see Figure \ref{fig_convex}).\qed
\end{proof}

The above observation gives us a method to locate a freeway efficiently:

\begin{theorem}\label{theo_FL1C1F}
Both variants of the 1C1F problem can be solved in $O(n\log n)$ time and $O(n)$ space.
\end{theorem}
\begin{proof}
We will solve the problem using a variation of the rotating calipers technique. This technique was already used in highway location problems in \cite{ahn09}. Although the metric considered was slightly simpler, the details are analogous. For simplicity in the explanation, we first consider the fixed length variant.

If the orientation of $h$ is known, we can compute the set $E_m$ and use Lemma \ref{lem_freekey}. Once this set is known, the problem can be solved in constant time. Since we not know the orientation of $h$, we will use rotating calipers. That is, we start with ${\alpha}=0$ (i.e., a horizontal line), rotate the line and keep track of the changes of the set $E_{\alpha}$.

Let $(\alpha_1,\alpha_2)$ be an interval of orientations in which set of $E_{\alpha}$ does not change (i.e., $E_{\alpha}=E_{\alpha'}$, for any ${\alpha},{\alpha'}\in(\alpha,\beta)$). For any such interval, the problem is of constant size, hence can be solved in constant time. In the following, we will show that the region $[0,\pi/2]$ can be partitioned into a small number of intervals in which the set $E_{\alpha}$ does not change.

Observe that, as ${\alpha}$ goes from $0$ to $\pi/2$, the associated coefficients $a_i$ and $b_i$ either remain constant (such as in regions $R_1$ and $R_2$) or are monotone. In particular, a point of $S$ can only be extreme with respect to the $i$-th region in an interval of orientations. That is, any single point of $S$ can generate a constant number of events in which the set $E_{\alpha}$ changes, hence the region $[0,\pi/2]$ will be split into at most $O(n)$ intervals. The most expensive part of the algorithm  (computationally speaking) is computing the convex hull, which takes $O(n\log n)$ time.

Finally we comment the case in which we are allowed to locate a freeway of variable length. Since increasing the length of the freeway can only decrease the travel time between two points, we can also assume that the highway has infinite length (i.e., we are locating a line instead of a segment). Thus, we simply fix a sufficiently large value $\ell$ and then execute the algorithm for the fixed length. \qed
\end{proof}


\begin{corollary}\label{cor_fixedor}
If the highway's orientation is fixed, both variants of the Free-1C1F problem can be solved in $O(n)$ time and $O(n)$ space.
\end{corollary}

We complete the problem by giving a matching lower bound:

\begin{theorem}\label{theo_lowerFL1C1F}
In the algebraic decision tree model both variants of the 1C1F need $\Omega(n\log n)$ time.
\end{theorem}
\begin{proof}
Proof of this claim follows from the fact that, in the particular case in which $\ell=v=\infty$, the solution to either variant is equivalent to finding the strip of minimum width that contains a given set of points. Since this problem is known to need $\Omega(n\log n)$ time \cite{lw-gcslp-86}, the same fact holds for the 1C1F problem. \qed
\end{proof}

\begin{center}
\begin{tabular}{|r|r|c|c|}
\hline
& length &  & $\max$ \\
\hline
\multirow{2}{24mm}{Turnpike} & fixed &  & $O(n^2)$ \\
\cline{2-4}
& free &  & $O(n\log n)$ \\
\cline{1-4}
\multirow{2}{24mm}{Freeway}& fixed &  & $\Theta(n\log n)$\\
\cline{2-4}
& free & & $\Theta(n\log n)$ \\\hline
\end{tabular}
\end{center}

\section{Concluding Remarks}\label{conclu}
In this paper we have considered several variants of the facility location problem introduced in~\cite{espejo11}.
Two models for the minmax criterion have been addressed,  the {\em turnpike} case, in which we only allow entering and leaving the highway at its
endpoints, and the  {\em freeway} case, in which one is allowed to enter and leave at any
point. All our results are summarized in Table 1.

The main open problem is to know whether or not the FL-1C1T variant can be solved in $o(n^2)$ time. One would think that this problem is similar to  {\em Rectilinear 2-center}~\cite{sergey,drezner87}. Unfortunately, the crucial property for solving the Rectilinear 2-center in
linear time (both covering boxes are contained in
the smallest enclosing axis-parallel rectangle of the
points) is not always true in our case. Refer to Figure~\ref{fig:counterexample}. A similar example can be done to show that the FL-1C1T-problem is not of LP-Type.

\begin{figure}[htb]
    \centering
    \includegraphics[width=0.35\textwidth]{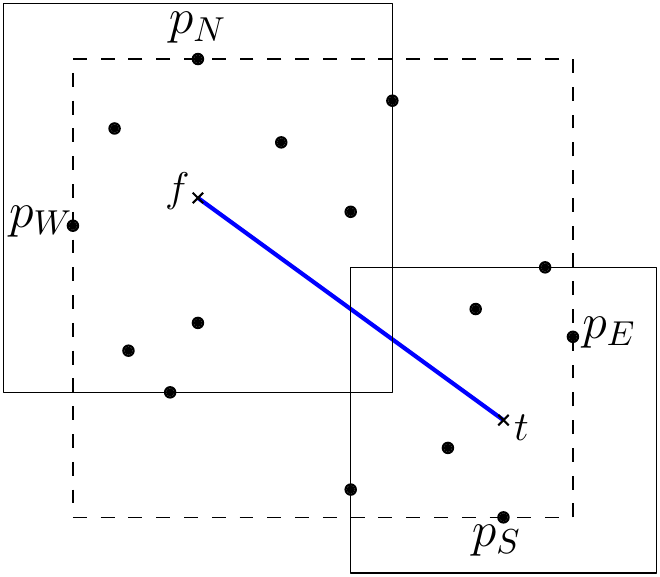}
    \caption{\small{Neither of the two squares is contained in the smallest enclosing
axis-parallel rectangle, and no extreme point lies on the boundary of their union.}}
    \label{fig:counterexample}
\end{figure}

\bibliography{carreteras}{}
\bibliographystyle{plain}
\end{document}